\documentclass{ifacconf}

\usepackage{graphicx}      
\usepackage{natbib}        
\usepackage{amsmath}
\usepackage{bm}
\usepackage{amssymb}
\usepackage{url}
\allowdisplaybreaks

\usepackage{tikz}
\usetikzlibrary{spy}
\usepackage[utf8]{inputenc}
\usepackage{pgfplots}
\usepgfplotslibrary{groupplots,dateplot}
\usetikzlibrary{patterns,shapes.arrows}
\pgfplotsset{compat=newest}
\def\axisdefaultheight{110pt}

\DeclareMathOperator*{\argmin}{arg\,min}

\newcounter{proposition}
\newenvironment{proposition}[1][]{%
	\refstepcounter{proposition}%
	\par\noindent\textit{Proposition~\theproposition. #1} \itshape
}{\par\normalfont}

\newenvironment{proof}[1][Proof]{%
	\par\noindent\textit{#1. }\normalfont
}{%
	\hfill$\square$\par
}

\theoremstyle{remark}
\newtheorem{remark}{Remark}
\begin{document}
\begin{frontmatter}

\title{Second-Order MPC-Based Distributed Q-Learning\thanksref{footnoteinfo}} 

\thanks[footnoteinfo]{This paper is part of a project that has received funding from the European Research Council (ERC) under the European Union’s Horizon 2020 research and innovation programme (Grant agreement No. 101018826
	- CLariNet).}

\author[First]{Samuel Mallick} 
\author[First]{Filippo Airaldi}
\author[First]{Azita Dabiri} 
\author[First]{Bart De Schutter}

\address[First]{Delft Center for Systems and Control, Delft University of Technology, Delft, The Netherlands (e-mail: {s.h.mallick, f.airaldi, a.dabiri, b.deschutter}@tudelft.nl).}

\begin{abstract}                
The state of the art for model predictive control (MPC)-based distributed Q-learning is limited to first-order gradient updates of the MPC parameterization.
In general, using second-order information can significantly improve the speed of convergence for learning and allowing the use of higher learning rates without introducing instability.
This work presents a second-order extension to MPC-based Q-learning with updates distributed across local agents, relying only on locally available information and neighbor-to-neighbor communication.
In simulation the approach is demonstrated to significantly outperform first-order distributed Q-learning in terms of learning speed.
\end{abstract}

\begin{keyword}
	Reinforcement Learning, Distributed Model Predictive Control, Distributed Learning, Q-Learning
\end{keyword}

\end{frontmatter}

\section{Introduction}
In recent years the paradigm of model predictive control (MPC)-based reinforcement learning (RL) has gained popularity as a data-driven optimal control technique \citep{gros_data-driven_2020}.
In this context, an MPC scheme is used as a function approximator for the RL value functions and policy.
RL methods are then used to learn the parameterization of the MPC scheme using value-based methods, e.g., Q-learning, or policy-based methods, e.g., policy gradient, with the parameterization adjusted such that the MPC scheme accurately approximates the optimal policy or RL value functions.
This methodology can be viewed as the use of RL to tune MPC components (e.g., prediction model, cost, and constraints) from closed-loop data, in order to optimize closed-loop performance.
Conversely, the approach can be considered as a model-based RL method.
Whereas in traditional deep RL, neural networks (NNs) are used as function approximators \citep{arulkumaran2017deep}, the parameterized MPC scheme provides a more interpretable replacement that additionally facilitates the wealth of theoretical results for MPC, e.g., stability and feasibility \citep{rawlings_model_2017}, to be used in an RL context.
This paradigm has successfully been applied in several application domains including home energy management \citep{cai2023learning}, traffic control \citep{airaldi2025reinforcement}, microgrid control \citep{cai2023energy}, and greenhouse climate control \citep{mallick2025reinforcement}.

However, for large-scale systems with distributed controllers the approach must be adapted for a multi-agent setting.
In these cases, control decisions for each agent are computed based on local information and with neighbor-to-neighbor communication (N2N), i.e., local communication with neighboring agents.
Furthermore, privacy is often a concern, with sensitive information not allowed to be shared between agents or a centralized unit.
For MPC, this case is addressed by a distributed MPC methodology, in which distributed optimization algorithms are used to resolve a globally formulated MPC optimization problem \citep{conte_distributed_2016}.
Equivalently for RL, the class of multi-agent RL (MARL) algorithms addresses multi-agent control systems, wherein a key challenge is the non-stationarity induced by agents learning contemporaneously \citep{busoniu_comprehensive_2008}.
Recent work \citep{mallick2024multi} introduced a variant of MPC-based Q-learning for the multi-agent case with linear systems, where a structured distributed MPC scheme is proposed as a distributed function approximator in RL, with the parameterization updated via a fully distributed learning algorithm.
The approach avoids the non-stationarity issue, proving the local updates to be equivalent to a global update under some assumptions.
However, the proposed distributed learning approach is restricted to first-order updates.
In general, second-order updates offer benefits over their first-order counterparts, as the use of second-order information can significantly improve the speed of convergence and assist in escaping saddle points by exploiting curvature \citep{nocedal2006numerical}.
Furthermore, whereas first-order methods typically require very small step sizes to avoid instability, which can be caused by jumping too far for a given search direction, second-order methods have a natural Newton step of one, and therefore allow the use of higher learning rates without destabilizing the learning \citep{nocedal2006numerical}.

In light of this, the current work presents an extension to MPC-based distributed Q-learning allowing for second-order updates.
Leveraging consensus algorithms, it is shown that a global second-order update can be decomposed into local updates relying only on local information and N2N communication.
The resulting distributed learning scheme enjoys superior convergence speed and stability with respect to the first-order variant.

The remainder of this article is organized as follows.
Section \ref{sec:prelim_back} introduces the problem setting and provides the relevant background theory from \cite{mallick2024multi}.
In Section \ref{sec:sceond_order}, the proposed second-order extension is presented.
Section \ref{sec:results} demonstrates the approach in simulation, while Section \ref{sec:conclusions} concludes the article.

\section{Preliminaries and Background}\label{sec:prelim_back}
\subsection{Problem Setting}
We define a multi-agent Markov decision process for $M$ agents, indexed with $i \in \mathcal{M} = \{1, \dots, M\}$, as the tuple
\begin{equation}
	\big(\{\mathcal{S}_i\}_{i \in \mathcal{M}}, \{\mathcal{A}_i\}_{i \in \mathcal{M}}, P, \{{L}_i\}_{i \in \mathcal{M}}, \mathcal{G}\big).
\end{equation}
The sets $\mathcal{S}_i$ and $\mathcal{A}_i$ are the local state and action sets, respectively, for agent $i$.
Furthermore, ${L}_i$ is the local cost function, while $P$ describes the state transition dynamics for the whole system. 
The graph $\mathcal{G} = (\mathcal{M}, \mathcal{E})$ defines a coupling topology between agents in the network, where edges $\mathcal{E}$ are ordered pairs $(i, j)$ indicating that agent $i$ may affect the cost and state transition of agent $j$. 
Define the neighborhood of agent $i$ as $\mathcal{N}_i = \{j \in {M} | (j, i) \in \mathcal{E}, i \neq j\}$.
Note that an agent is not in its own neighborhood, i.e., $i \notin \mathcal{N}_i$.
We assume the graph $\mathcal{G}$ is connected.
Additionally, agents $i$ and $j$ can communicate if $i \in \mathcal{N}_j$ or $j \in \mathcal{N}_i$.
We consider agents to be linear dynamical systems with state $s_i \in \mathcal{S}_i \subseteq \mathbb{R}^n$ and control input $a_i \in \mathcal{A}_i \subseteq \mathbb{R}^m$.
However, the true dynamics $P$ of the network are assumed to be unknown.

We consider a cooperative RL setting where agents aim to minimize a common goal; however, information exchange is limited to agents' neighborhoods, and privacy is considered such that no sharing of local functions, e.g., cost functions or dynamics, is permitted.
At time step $t$, agent $i$ observes its own state $s_{i, t} \in \mathcal{S}_i$ and takes an action with a local policy parametrized by the local parameter vector $\theta_i \in \mathbb{R}^{n_{\theta_i}}$, $a_{i, t} = \pi_{i, \theta_i}(s_{i, t})$, observing the incurred local cost $L_{i}(s_{i, t}, a_{i, t})$ and the next state $s_{i, t+1}$.
We define the global parametrization as $\theta = [\theta_1^\top,...,\theta_M^\top]^\top \in \mathbb{R}^{n_\theta}$, with $n_\theta = \sum_{i\in \mathcal{M}} n_{\theta_i}$.
The joint policy, parametrized by $\theta$, is then $\pi_\theta(s) = \{\pi_{i, \theta_i}(s_i)\}_{i \in \mathcal{M}}$, with $a \in \mathcal{A} = \mathcal{A}_1 \times \dots \times \mathcal{A}_M$ the joint action $a = \{a_i\}_{i \in \mathcal{M}} = \pi_\theta(s)$, where $s \in \mathcal{S} = \mathcal{S}_1 \times \dots \times \mathcal{S}_M$ is the joint state $s = \{s_i\}_{i \in \mathcal{M}}$.
The cooperative goal is to minimize, by modifying the parameters $\theta_i$, the discounted cost over an infinite horizon 
\begin{equation}
	\label{eq:RL_goal}
	J(\pi_\theta) = \mathbb{E} \left[ \sum_{t=0}^{\infty} \gamma^t L\big(s_t, \pi_\theta(s_t)\big)\right],
\end{equation}
with $\gamma \in (0, 1]$ and
\begin{equation}
L(s_t, a_t) = \frac{1}{M}\sum_{i \in \mathcal{M}}L_{i}(s_{i, t}, a_{i, t})
\end{equation}
the average of the agents' local costs.
The expectation in \eqref{eq:RL_goal} is over the stochastic state trajectory, determined by the distribution of initial states $s_0$ and the potentially stochastic dynamics $P$.
\subsection{First-Order Distributed MPC-based Q-Learning}\label{sec:first_order}
For a given policy $\pi$, the familiar RL notion of the value function \citep{sutton_reinforcement_2018} is defined as
\begin{equation}
	V_{\pi}(s_t) = \mathbb{E}\bigg[\sum_{\tau = t}^\infty \gamma^{\tau - t}L\big(s_\tau, \pi(s_\tau)\big)\Big| s_{\tau|{\tau = t}} = s_t \bigg],
\end{equation}
where the expectation is again over the stochastic state trajectory starting from $s_t$, and the action-value function as
\begin{equation}
	Q_{\pi}(s_t, a_t) = L(s_t, a_t) + V_{\pi}(s_{t+1}).
\end{equation}
Value-based RL methods in general attempt to learn an approximation $Q_\theta$ of the optimal action-value function 
\begin{equation}
	Q_\theta(s, a) \approx Q^\star(s, a) = \min_{\pi}Q_{\pi}(s, a), \: \forall (s, a) \in \mathcal{S} \times \mathcal{A}
\end{equation}
and then infer the policy as 
\begin{equation}
	\pi_\theta(s) = \argmin_a Q_\theta(s, a).
\end{equation}
While in deep RL methods an NN is used as parametric function approximator for $Q_\theta$, in MPC-based RL the NN is replaced by an MPC scheme.

In \cite{mallick2024multi}, this idea was extended to the multi-agent case, with the use of a structured convex distributed MPC scheme, parameterized in $\theta$, proposed as distributed function approximator.
Consider the action-value function approximation
\begin{subequations}\label{eq:Q}
	\begin{align}
		Q_\theta(s, a) = &\min_{\substack{(\{\bm{x}_i\}, \{\bm{u}_i\})_{i \in \mathcal{M}}}} \sum_{i \in \mathcal{M}} F_{\theta_i}\big(\bm{x}_i, \{\bm{x}_j\}_{j \in \mathcal{N}_i}, \bm{u}_i\big) \\
		\text{s.t.}& \quad \forall  i \in \mathcal{M}: \nonumber \\
		&u_i(0) = a_i, \: x_i(0) = s_i \\
		&H_{\theta_i}\big(\bm{x}_i, \{\bm{x}_j\}_{j \in \mathcal{N}_i}, \bm{u}_i\big) \leq 0 \label{eq:ineq}\\
		&x_i(k+1)= A_{i, \theta_i} x_i(k) + \sum_{j \in \mathcal{N}_i} A_{ij, \theta_i} x_j(k)  \nonumber\\
		 &\quad + B_{i, \theta_i} u_i(k) + b_{\theta_i}, \: k = 0,\dots,N-1. \label{eq:dynamics}
	\end{align}
\end{subequations}
The convex functions $F_{\theta_i}$ are local costs, taking as arguments the state trajectories of an agent and its neighbors, and the action trajectory of the given agent, over a prediction horizon of $N$ steps,
\begin{equation}
	\begin{aligned}
		\bm{x}_i &= \big[x_i^\top(0), \dots, x_i^\top(N)\big], \\
		\bm{u}_i &= \big[u_i^\top(0), \dots, u_i^\top(N-1)\big].
	\end{aligned}
\end{equation}
Constraint \eqref{eq:dynamics} is an approximation of the agents' true dynamics $P$, while constraint \eqref{eq:ineq} contains general coupled linear inequality constraints.
Through learning, the parameterization of the cost, model, and constraints can be modified to improve the global closed-loop performance.
Note that slack variables, penalized in the cost, can soften the inequality constraints such that feasibility is maintained while learning $\theta$; however, for simplicity they are omitted.
See \cite{mallick2024multi, gros_data-driven_2020} for more details. 

The value function and policy are then constructed from the same MPC scheme as
\begin{equation}
	\label{eq:policy}
	\begin{aligned}
		V_\theta(s) = &\argmin_{\{u_i(0)\}_{i \in \mathcal{M}}} \: \sum_{i \in \mathcal{M}} F_{\theta_i}\big(\bm{x}_i, \{\bm{x}_j\}_{j \in \mathcal{N}_i}, \bm{u}_i\big)\\
		\text{s.t.}& \quad \forall  i \in \mathcal{M}: \\ 
		&x_i(0) = s_i, \:\text{\eqref{eq:ineq}, \eqref{eq:dynamics}},
	\end{aligned}
\end{equation}
and $\pi_\theta(s) = \{u_i^\star(0)\}_{i \in \mathcal{M}}$, satisfying the fundamental Bellman equations \citep{sutton_reinforcement_2018}.
For clarity in the following, all equality constraints in \eqref{eq:Q} are grouped in 
\begin{equation}
	G_{\theta_i}\big(s_i, a_i, \bm{x}_i, \{\bm{x}_j\}_{j \in \mathcal{N}_i}, \bm{u}_i\big) = 0.
\end{equation}

The parameterization $\theta$ of the distributed MPC scheme can then be learned with the value-based RL method, Q-learning.
With the TD error signal at time step $t$
\begin{equation}
	\delta_t = L(s_t, a_t) + \gamma V_\theta(s_{t+1}) - Q_\theta(s_t, a_t),
\end{equation}
First-order Q-learning updates $\theta$ as
\begin{equation} \label{eq:fo_update}
	\theta \gets \theta + \alpha 
	\frac{1}{T}\sum_{\tau \in \mathcal{T}} \delta_\tau \nabla_\theta Q_\theta(s_\tau, a_\tau),
\end{equation}
for the TD error $\delta_\tau$ and the gradient $\nabla_\theta Q_\theta(s_\tau, a_\tau)$ at the time steps $\mathcal{T} = \{t_1, \dots, t_T\}$, randomly sampled from a replay buffer of stored transitions.
For $T = 1$, \eqref{eq:fo_update} reduces to classical recursive Q-learning, where the most recent observed data is used to update $\theta$ and then discarded.
In general, $T > 1$ provides a better sample average of the true gradient and is desired for stability of the learning \citep{lin1992self}.

It is shown in \cite{mallick2024multi} that, via the alternating direction method of multipliers (ADMM) \citep{boyd_distributed_2010} and global average consensus (GAC) \citep{olfati-saber_consensus_2007}, the policy and value functions can be evaluated distributively by agents using only local information and N2N communication, i.e., $\pi_{i, \theta_i}(s_i) = u^\star_i(0)$, $V_\theta(s)$, $Q_\theta(s, a)$, and $\delta$ are made available locally. 
Via sensitivity analysis \citep{buskens_sensitivity_2001}, the gradient of the Q-function coincides with that of the Lagrangian of \eqref{eq:Q}, $\mathcal{L}_{\theta}$, evaluated at the optimal primal and dual variables for \eqref{eq:Q}, stacked in the vector $p^\star$.
That is,
\begin{equation}
	\label{eq:Q_sens}
	\nabla_\theta Q_\theta(s, a) = \frac{\partial \mathcal{L}_\theta(s, a, p^\star)}{\partial \theta}.
\end{equation}
The Lagrangian, due to the structure of \eqref{eq:Q}, is the sum of local Lagrangians
\begin{equation}
	\begin{aligned}
		\mathcal{L}_\theta(s, a, p) &= \sum_{i \in \mathcal{M}} \mathcal{L}_{\theta_i}(s_i, a_i, p_i)\\ 
		&= \sum_{i \in \mathcal{M}} F_{\theta_i}\big(\bm{x}_i, \{\bm{x}_j\}_{j \in \mathcal{N}_i}, \bm{u}_i\big) \\
		&\quad+ \bm{\lambda}_i^\top G_{\theta_i}\big(s_i, a_i, \bm{x}_i, \{\bm{x}_j\}_{j \in \mathcal{N}_i}, \bm{u}_i\big) \\
		&\quad+ \bm{\mu}_i^\top H_{\theta_i}\big(\bm{x}_i, \{\bm{x}_j\}_{j \in \mathcal{N}_i}, \bm{u}_i\big),
	\end{aligned}
\end{equation}
where $p_i$ is the locally relevant set of primal and dual variables
\begin{equation}
	p_i = \big(\bm{x}_i, \bm{u}_i, \bm{\sigma}_i, \{\bm{x}_j\}_{j \in \mathcal{N}_i}, \bm{\lambda}_i, \bm{\mu}_i\big).
\end{equation} 
In \cite{mallick2024multi}, it is shown that the optimal set of primal and duals $p_i^\star$ is available locally following the ADMM procedure, and that the global update \eqref{eq:fo_update} decomposes into the local updates
\begin{equation}
	\label{eq:fo_local}
	\theta_i \gets \theta_i + \alpha \frac{1}{T} \sum_{\tau \in \mathcal{T}}\delta_\tau \frac{\partial \mathcal{L}_{\theta_i}(s_{i, \tau}, a_{i, \tau}, p_i^\star)}{\partial \theta_i},\: i \in \mathcal{M},
\end{equation}
as $\delta_\tau$ has been made available locally via consensus.

The result is a fully distributed MPC-RL scheme where agents utilize only local information and N2N communication during both training and deployment.
The updates \eqref{eq:fo_local} are, however, first order.
In Section \ref{sec:sceond_order} we extend the methodology to allow distributed second-order updates.

\subsection{Global Average Consensus}
Here we briefly introduce the GAC algorithm, as it is used explicitly in the sequel.
The GAC algorithm allows a network of agents to agree on the sum\footnote{Usually GAC is formulated for agreement on the average of local variables; however, provided the number of agents $M$ is known locally, the sum follows trivially.} of local variables $v_{i} \in \mathbb{R}, i \in \mathcal{M}$, communicating over the graph $\mathcal{G}$.
For each agent, the algorithm updates values as 
\begin{equation}\label{eq:cons_iter}
	v_i^{(\tau+1)} = [\bm{P}]_{i,i} v_i^{(\tau)} + \sum_{j \in \mathcal{N}_i} [\bm{P}]_{i,j} v_j^{(\tau)},
\end{equation}
for $v_i^{(0)} = M v_i$, where $\bm{P} \in \mathbb{R}_+^{M \times M}$ is a doubly stochastic matrix, i.e., entries in each row and column sum to 1, where $[\bm{P}]_{i, j}$ is the entry of $\bm{P}$ in the $i$th row and $j$th column.
The iterates converge as \citep{olfati-saber_consensus_2007}
\begin{equation}
	\lim_{\tau \to \infty} v_i^{\tau} = \frac{1}{M}\sum_{j \in \mathcal{M}} v_j^{(0)} = \sum_{j \in \mathcal{M}} v_{j},\: j \in \mathcal{M}.
\end{equation}

\section{Second-Order Learning}\label{sec:sceond_order}
For Q-learning, a global second-order update of $\theta$ is
\begin{equation}\label{eq:global_update}
	\theta \gets \theta - \alpha \bm{d},
\end{equation}
 where $\bm{d}$ is the solution to 
\begin{equation}\label{eq:Hug}
	(\bm{H} + \bm{\Lambda})\bm{d} = \bm{q}
\end{equation}
with \citep{airaldi2025reinforcement}
\begin{equation}
	\small
	\begin{aligned}
		\bm{q} &= -\frac{1}{T}\sum_{\tau \in \mathcal{T}} \delta_\tau \nabla_\theta Q_\theta(s_\tau, a_\tau) \\
		\bm{H} &= \frac{1}{T}\sum_{\tau \in \mathcal{T}} \bigg( \nabla_\theta Q_\theta(s_\tau, a_\tau) \nabla_\theta Q_\theta^\top(s_\tau, a_\tau) - \delta_\tau \nabla_\theta^2Q_\theta(s_\tau, a_\tau) \bigg),
	\end{aligned}
\end{equation} 
where $\bm{\Lambda}$ is added as a regularization, providing a well-behaved update, e.g., rendering $\bm{H} + \bm{\Lambda}$ non-singular or positive definite.
The form and computation of $\bm{\Lambda}$ is discussed in what follows.
As it stands, note that this solution does not trivially separate over agents' local information due to the cross-coupling in $\bm{H}$.
In the following, we show how leveraging consensus on a specific set of local variables allows a second-order update to decouple.

\subsection{Recursive Update}
For simplicity, and to give intuition, we first elaborate the recursive case $T = 1$ and where \eqref{eq:Q} contains only a first-order parameterization, i.e., $\nabla_\theta^2Q_\theta(s_t, a_t) = \bm{0}$. 
Define the vector of all first-order information at time step $t$, $\bm{g}_t \in \mathbb{R}^{n_\theta}$, as
\begin{equation}
	\begin{aligned}
		\bm{g}_t &= \nabla_\theta Q_\theta(s_t, a_t)
		= \begin{bmatrix} \bm{g}_{1, t}^\top & \dots & \bm{g}_{M, t}^\top
		\end{bmatrix}^\top,
	\end{aligned}
\end{equation}
where $\bm{g}_{i, t} = \partial \mathcal{L}_{\theta_i}(s_{i, t}, a_{i, t}, p_{i, t}^\star)/\partial \theta_i \in \mathbb{R}^{n_{\theta_i}}$.
Thus, we are interested in finding a solution to \eqref{eq:Hug} in the form 
\begin{equation}\label{eq:recursive_update}
	(\bm{g}_t \bm{g}_t^\top + \bm{\Lambda})\bm{d} = -\delta_t \bm{g}_t.
\end{equation} 
Observe that $\bm{g}_t \bm{g}_t^\top$ is a rank-1 matrix update of $\bm{\Lambda}$.
As such, assuming $\bm{g}_t \bm{g}_t^\top + \bm{\Lambda}$ is non-singular, the inverse can be computed using the Sherman-Morrison formula \cite{horn_matrix_2012} as 
\begin{equation}
	\label{eq:Sherman_Morrison}
	\begin{aligned}
		(\bm{g}_t \bm{g}_t^\top + \bm{\Lambda})^{-1} &= \bm{\Lambda}^{-1} - \frac{(\bm{\Lambda})^{-1}\bm{g}_t \bm{g}_t^\top (\bm{\Lambda})^{-1}}{1 + \bm{g}_t^\top (\bm{\Lambda})^{-1} \bm{g}_t} 
	\end{aligned}
\end{equation}
Note that, as $\bm{g}_t \bm{g}_t^\top \succeq 0$, we have that $(\bm{g}_t \bm{g}_t^\top + \sigma \bm{I})^{-1}$ exists for any scalar $\sigma > 0$.
Hence, $\bm{\Lambda}$ can be fixed as $\sigma \bm{I}$ \textit{a priori}, and the update \eqref{eq:recursive_update} can be expressed as 
\begin{equation}
	\begin{aligned}
		\bm{d} &= -\Big(\frac{1}{\sigma}\bm{I} - \frac{\bm{g}_t \bm{g}_t^\top}{\sigma^2 + \sigma \|\bm{g}_t\|^2}\Big) \delta_t \bm{g}_t =\frac{-\delta_t}{\sigma + \|\bm{g}_t\|^2} \bm{g}_t.
	\end{aligned}
\end{equation}
Therefore, the global solution can be decomposed into local solutions $\bm{d} = [\bm{d}_1^\top, \dots, \bm{d}_M^\top]^\top$ with 
\begin{equation}
	\label{eq:local_recursive_update}
	\bm{d}_i = \frac{-\delta_t}{\sigma + \|\bm{g}_t\|^2} \bm{g}_{i, t},
\end{equation}
with the local parameter update therefore $\theta_i \gets \theta_i - \alpha \bm{d}_i$.
As $\|\bm{g}_t\|^2 = \sum_{i \in \mathcal{M}} \|\bm{g}_{i, t}\|^2$ is a sum of locally known values, this scalar value can be made available locally, alongside the TD error $\delta_t$, via the GAC algorithm.

\subsection{Full Second-Order Update}
We now provide a result that shows how, leveraging consensus on additional variables, a decoupled update is possible for $T > 1$ and $\nabla_\theta^2Q_\theta(s_t, a_t) \neq \bm{0}$.
Given the structure of \eqref{eq:Q}, observe that
\begin{equation}
	\frac{\partial \mathcal{L}_{\theta_i}}{\partial \theta_j \partial \theta_i} = \bm{0}, \: \forall i, j \: , \: i \neq j.
\end{equation}
Thus, $\nabla_\theta^2Q_\theta(s_t, a_t)$ is block diagonal.
Define $\bm{K} \in \mathbb{R}^{n_\theta \times n_\theta}$ as
\begin{equation}
	\bm{K} = \sum_{\tau \in \mathcal{T}}\delta_\tau \nabla_\theta^2Q_\theta(s_\tau, a_\tau) = \begin{bmatrix}
		\bm{K}_{1} & \dots & \bm{0} \\
		\vdots & \ddots & \vdots \\
		\bm{0} & \dots & \bm{K}_{M}
	\end{bmatrix},
\end{equation}
where $\bm{K}_{i}  = \sum_{\tau \in \mathcal{T}}\delta_\tau \partial^2 \mathcal{L}_{\theta_i}(s_{i, \tau}, a_{i, \tau}, p_{i, \tau}^\star) / \partial \theta_i^2\in \mathbb{R}^{n_{\theta_i} \times n_{\theta_i}}$.
We are now interested in a solution to \eqref{eq:Hug} in the form 
\begin{equation}\label{eq:replay_update}
	\Bigg(\sum_{\tau \in \mathcal{T}} \bm{g}_\tau \bm{g}_\tau^\top - \bm{K} + T \bm{\Lambda}\Bigg) \bm{d} = -\sum_{t \in \mathcal{T}}\delta_\tau \bm{g}_\tau.
\end{equation}
Define the matrices
\begin{equation}
	\begin{aligned}
		\tilde{\bm{K}}_{i} &= (T \sigma_i \bm{I} - \bm{K}_i)^{-1} \in \mathbb{R}^{n_{\theta_i} \times n_{\theta_i}}, \\
		\tilde{\bm{K}} &= \begin{bmatrix}
			\tilde{\bm{K}}_{1} & \dots & \bm{0} \\
			\vdots & \ddots & \vdots \\
			\bm{0} & \dots & \tilde{\bm{K}}_{M}
		\end{bmatrix} \in \mathbb{R}^{n_\theta \times n_\theta}, \\
		 \bm{C} &= \begin{bmatrix}
			\bm{g}_{t_1}^\top \tilde{\bm{K}} \bm{g}_{t_1} & \bm{g}_{t_1}^\top \tilde{\bm{K}} \bm{g}_{t_2} & \dots & \bm{g}_{t_1}^\top \tilde{\bm{K}} \bm{g}_{t_T} \\
			\bm{g}_{t_2}^\top \tilde{\bm{K}} \bm{g}_{t_1} & \bm{g}_{t_2}^\top \tilde{\bm{K}} \bm{g}_{t_2} & \dots & \bm{g}_{t_2}^\top \tilde{\bm{K}} \bm{g}_{t_T} \\
			\vdots & \vdots & \ddots & \vdots \\
			\bm{g}_{t_T}^\top \tilde{\bm{K}} \bm{g}_{t_1} & \bm{g}_{t_T}^\top \tilde{\bm{K}} \bm{g}_{t_2} & \dots & \bm{g}_{t_T}^\top \tilde{\bm{K}} \bm{g}_{t_T}
		\end{bmatrix} \in \mathbb{R}^{T \times T}, \\
		\bm{\delta} &= [\delta_{t_1}, \dots,  \delta_{t_T}]^\top \in \mathbb{R}^T, \\
		\bm{\Lambda} &= \begin{bmatrix}
			\sigma_{1} \bm{I} & \dots & \bm{0} \\
			\vdots & \ddots & \vdots \\
			\bm{0} & \dots & \sigma_{M} \bm{I}
		\end{bmatrix} \in \mathbb{R}^{n_\theta \times n_\theta},
	\end{aligned}
\end{equation}
The terms $\sigma_i > 0$ in $\bm{\Lambda}$ are chosen as local regularization terms.
Recall, from Section \ref{sec:first_order}, that at each time step $\delta_t$ is made available locally, such that each agent can construct $\bm{\delta}$.
Furthermore, observe that, as $\tilde{\bm{K}}$ is block diagonal, each element of $\bm{C}$ can be expressed as
\begin{equation}
	[\bm{C}]_{\tau, \tau^\prime} = \sum_{i \in \mathcal{M}} \bm{g}_{i, \tau}^\top \tilde{\bm{K}}_i \bm{g}_{i, \tau^\prime}.
\end{equation}
As a sum of locally known values, these entries can be made available to all agents via the GAC algorithm, such that $\bm{C}$ is available locally.
Additionally, as $\bm{C}$ is symmetric, consensus on only $T(T+1)/2$ scalar values is required.
Finally, define $\bm{G}_i \in \mathbb{R}^{n_{\theta_i} \times T}$ as 
\begin{equation}
	\bm{G}_i = \begin{bmatrix}
		\bm{g}_{i, t_1}, \dots, \bm{g}_{i, t_T}
	\end{bmatrix}.
\end{equation}
We now give a result providing a solution to \eqref{eq:replay_update} that separates across agents.
\begin{proposition}\label{prop:1}
	Assume that $\sigma_i$, $i = 1,\dots,M$, are chosen such that the following inverses exist:
	\begin{equation}
		\bigg(\sum_{\tau \in \mathcal{T}} \bm{g}_\tau \bm{g}_\tau^\top - \bm{K} + T \bm{\Lambda}\bigg)^{-1}, (\bm{I} + \bm{C})^{-1}, (T \sigma_i \bm{I} - \bm{K}_i)^{-1}.
	\end{equation}
	The solution to \eqref{eq:replay_update} is $\bm{d} = \big[\bm{d}_1^\top, \dots, \bm{d}_M^\top\big]^\top$ with
	\begin{equation}
		\label{eq:local_so_update}
		\bm{d}_i = -\tilde{\bm{K}}_i \bm{G}_i \Big(\bm{\delta} - (\bm{I} + \bm{C})^{-1}\bm{C}\bm{\delta}\Big).
	\end{equation} 
\end{proposition}
\begin{proof}
	Define the matrices $\bm{A} \in \mathbb{R}^{n_\theta \times n_\theta} $ and $\bm{U} \in \mathbb{R}^{n_\theta \times T}$ as
	\begin{equation}
		A = T \bm{\Lambda} - \bm{K}, \quad \bm{U} = \begin{bmatrix}
			\bm{g}_{t_1} & \dots & \bm{g}_{t_T}
		\end{bmatrix},
	\end{equation}
	such that \eqref{eq:replay_update} can be written
	\begin{equation}
		(\bm{U} \bm{U}^\top + \bm{A})\bm{d} = -\bm{U} \bm{\delta}.
	\end{equation}
	As $\bm{K}$ and $\bm{\Lambda}$ are block diagonal, $\bm{A}$ is also block diagonal, and the inverse is formed from the inverse of each block, such that $\bm{A}^{-1} = \tilde{\bm{K}}$.
	This inverse exists as $\tilde{\bm{K}}_i$ for $i \in \mathcal{M}$ is non-singular by assumption.
	Observe that $\bm{U} \bm{U}^\top$ is (at most) a rank $\min(n_\theta, T)$ matrix update of $\bm{A}$, such that, leveraging the Woodbury matrix identity \citep{horn_matrix_2012}, the inverse can be written as
	\begin{equation}
		(\bm{U} \bm{U}^\top + \bm{A})^{-1} = \bm{A}^{-1} - \bm{A}^{-1}\bm{U}(\bm{I} + \bm{U}^\top \bm{A}^{-1} \bm{U})^{-1}\bm{U}^\top \bm{A}^{-1}.
	\end{equation}
	The solution to \eqref{eq:replay_update} can then be expressed
	\begin{equation}
		\bm{d} = -\tilde{\bm{K}}\bm{U}\bm{\delta} + \tilde{\bm{K}}\bm{U}(\bm{I} + \bm{U}^\top \tilde{\bm{K}}\bm{U})^{-1} \bm{U}^\top \tilde{\bm{K}} \bm{U} \bm{\delta}.
	\end{equation} 
	Observing that $\bm{U}^\top \tilde{\bm{K}}\bm{U} = \bm{C}$ by definition, this can be written
	\begin{equation}\label{eq:ppp}
		\bm{d} = -\tilde{\bm{K}}\bm{U}\big(\bm{\delta} - (\bm{I} + \bm{C})^{-1}\bm{C}\bm{\delta}\big),
	\end{equation}
	where the remaining inverse exists by assumption.
	Finally, as $\bm{U}$ can equivalently be written
	\begin{equation}
		\bm{U} = \begin{bmatrix}
			\bm{G}_1^\top & \dots & \bm{G}_M^\top
		\end{bmatrix}^\top
	\end{equation}
	and $\tilde{\bm{K}}$ is block diagonal, the block $\tilde{\bm{K}}_i$ multiplies the relevant component $\bm{G}_i$, and the result follows.
\end{proof}
The distributed second order update is then
\begin{equation}\label{eq:so_local}
	\theta_i \gets \theta_i + \alpha \bigg(\tilde{\bm{K}}_i \bm{G}_i \Big(\bm{\delta} - (\bm{I} + \bm{C})^{-1}\bm{C}\bm{\delta}\Big)\bigg).
\end{equation}

For $T=1$, $\nabla_\theta^2 Q_\theta = \bm{0}$, and $\bm{\Lambda} = \sigma \bm{I}$ we have $\tilde{\bm{K}}_i = \frac{1}{\sigma} \bm{I}$, $\bm{G}_i = \bm{g}_{i, t_1}$, $\bm{\delta} = \delta_{t_1}$, and $\bm{C} = \frac{1}{\sigma}\|\bm{g}_{t_1}\|^2$.
The update \eqref{eq:local_so_update} then reduces to
\begin{equation}
	\begin{aligned}
		\bm{d}_i &= -\frac{1}{\sigma}\bm{g}_{i, t_1}\bigg(\delta_{t_1} - \Big(1 + \frac{1}{\sigma}\|\bm{g}_{t_1}\|^2\Big)^{-1}\frac{1}{\sigma}\|\bm{g}_{t_1}\|^2 \delta_{t_1}\bigg) \\
		&=\frac{-\delta_{t_1}}{\sigma + \|\bm{g}_{t_1}\|^2}\bm{g}_{i, t_1},
	\end{aligned}
\end{equation}
and the recursive update \eqref{eq:local_recursive_update} is recovered.
\begin{remark}
	The existence of the inverses assumed for Proposition \ref{prop:1} can be guaranteed using only local information by selecting $\sigma_i$ such that $T \sigma_i \bm{I} - \bm{K}_i \succ 0$.
	Then, as this choice renders $\tilde{\bm{K}}$ positive definite, we have that $\bm{C} \succeq 0$, and $\bm{I}+\bm{C}$ is positive definite and thus invertible.
	Finally, with positive definite diagonal blocks, $\bm{A} \succ 0$, and, as $\bm{U}\bm{U}^\top \succeq 0$, the matrix $\bm{U}\bm{U}^\top + \bm{A}$, or equivalently $\sum_{\tau \in \mathcal{T}} \bm{g}_\tau \bm{g}_\tau^\top - \bm{K} + T \bm{\Lambda}$, is positive definite and invertible.
	This then corresponds to a regularization for positive definiteness of the Hessian, which is a common choice for second-order methods.
	In the case that regularization only for non-singularity is desired, note that $\sigma_i$ can be chosen locally such that $T \sigma_i \bm{I} - \bm{K}_i$ is non-singular.
	Then $\bm{I} + \bm{C}$ becomes available locally following consensus, and can be checked for invertibility, such that the local update \eqref{eq:local_so_update} can be computed.
	In the case that this choice results in singular $\sum_{\tau \in \mathcal{T}} \bm{g}_\tau \bm{g}_\tau^\top - \bm{K} + T \bm{\Lambda}$, the interpretation is the distributed update \eqref{eq:local_so_update} being computable but not equivalent to a centralized update, which is not well-defined.
\end{remark}
\begin{remark}
	The local update \eqref{eq:local_so_update} requires an agent to compute the numerical inverses of the two symmetric matrices $T \sigma_i \bm{I} - \bm{K}_i$ and $\bm{I} + \bm{C}$, of size $n_{\theta_i}\times n_{\theta_i}$ and $T \times T$, respectively.
	Hence, the computational burden increases as the parameterization becomes more complex and the replay buffer larger, but is independent of the network size $M$.
	In contrast, a centralized update requires the inverse of the full Hessian of size $\sum_{i \in \mathcal{M}} n_{\theta_i} \times \sum_{i \in \mathcal{M}} n_{\theta_i}$, which is independent of $T$ but increases with $M$.
\end{remark}
\begin{remark}
	The first-order update \eqref{eq:fo_local} includes consensus on the three scalars $V_\theta(s)$, $Q_\theta(s, a)$, and $L(s, a)$ at each time step, for agents to have local knowledge of $\delta$.
	The second-order update \eqref{eq:local_so_update} additionally introduces consensus on $T(T+1)/2$ scalars at each parameter for local knowledge of $\bm{C}$.
	All values can be stacked in a vector and GAC applied for agreement on the vector, such that the number of iterations of \eqref{eq:cons_iter} is constant.
	The amount of data communicated between agents scales as $\mathcal{O}(T^2)$, but is, however, independent of the network size $M$.
\end{remark}

\section{Results}\label{sec:results}
This section presents a numerical example. 
Source code and simulation results can be found at \url{https://github.com/SamuelMallick/dmpcrl_so}.
We consider a three-agent system with state coupling in a chain, i.e., $1 \leftrightarrow 2 \leftrightarrow 3$, with the real (unknown) dynamics 
\begin{equation}
	s_i(t+1) = A_i s_i(t)+ \sum_{j \in \mathcal{N}_i} A_{ij} s_j(t) + B_i a_i(t)  + [e_i(t), 0]^\top
\end{equation}
where
\begin{equation} \label{eq:true_mod}
	A_i = \begin{bmatrix}
		0.9 & 0.35 \\ 0 & 1.1
	\end{bmatrix}, \{A_{ij}\}_{j \in \mathcal{N}_i} = \begin{bmatrix}
		0 & 0 \\ 0 & -0.1
	\end{bmatrix}, B = \begin{bmatrix}
	0.0813 \\ 0.2
	\end{bmatrix},
\end{equation}
and $e_i(t)$ uniformly distributed over the interval $[-0.1, 0]$.
The task is to drive the states towards the origin while avoiding violations of the constraints $\underline{s} \leq s_i \leq \overline{s}$, with $\underline{s} = [0, -1]^\top$ and $\overline{s} = [1, 1]^\top$.
Local costs are
\begin{equation}
	\begin{aligned}
		&L_{i}(s_i, a_i) = \|s_i\|_2^2 + \frac{1}{2}\|a_i\|_2^2 \\
		&{} \: + \max\big\{0, \omega^\top (\underline{s} - s_i)\big\} + \max\big\{0, \omega^\top (s_i - \overline{s})\big\},
	\end{aligned} 
\end{equation}
with $\omega = [500, 500]^\top$.
Non-positive noise on the first state biases that term towards violating the lower bound of zero.
Hence, agents must learn to regulate the state close to zero while accounting for the noise.
The true model is unknown, with agents instead knowing the incorrect initial model
\begin{equation}
	\label{eq:inac_mod}
	\begin{aligned}
		&\Hat{A}_i = \begin{bmatrix}
			a_{i,11} & a_{i, 12}\\ 0 & a_{i, 22}
		\end{bmatrix} = \begin{bmatrix}
			1 & 0.25\\ 0 & 1
		\end{bmatrix}, \\
		 &\Hat{A}_{ij} = \begin{bmatrix}
			0 & 0 \\ 0 & a_{ij}
		\end{bmatrix} = \begin{bmatrix}
		0 & 0 \\ 0 & 0
		\end{bmatrix}, \: j \in \mathcal{N}_i \\
		&\Hat{B}_i = \begin{bmatrix}
			b_{i,1} \\ b_{i,2}
		\end{bmatrix} = \begin{bmatrix}
		0.0312 \\ 0.25
		\end{bmatrix}.
	\end{aligned}
\end{equation}
We implement the following distributed MPC scheme:
\begin{equation*}
	\begin{aligned}\label{eq:academic_dmpc}
		\min_{\substack{\{(\bm{x}_i, \bm{u}_i, \bm{\sigma}_i)\}_{i \in \mathcal{M}}}} &\sum_{i \in \mathcal{M}}\Bigg( V_{i, 0} + \sum_{k \in \mathcal{K}} f_i^\top \begin{bmatrix}
			x_i(k) \\ u_i(k)
		\end{bmatrix} \\
		&+ \frac{1}{2}\gamma^k \Big( \big(Q^\top_i x_i(k)\big)^2 + \big(R^\top_i u_i(k)\big)^2 \\
		&+ \omega_i^\top \sigma_i(k) \Big) \Bigg)\\
		\text{s.t.}& \quad \forall  i \in \mathcal{M}, \: \forall k \in \mathcal{K}: \nonumber \\
		&x_i(k+1) = \Hat{A}_i x_i(k) + \Hat{B}_i u_i(k) \\
		&\quad\quad+ \sum_{j \in \mathcal{N}_i} \Hat{A}_{ij} x_j(k) + b_i\\
		&\underline{s} + \underline{x}_i - \sigma_i(k) \leq x_i(k) \leq \overline{s} + \overline{x}_i + \sigma_i(k)\\
		&-1 \leq u_i(k) \leq 1, \: x_i(0) = s_i
	\end{aligned}
\end{equation*}
where $M = 3$, $\mathcal{K} = \{0,\dots,9\}$, and $\gamma = 0.9$.
The learnable parameters for each agent are then
\begin{equation}
	\begin{aligned}
		\theta_i = &(V_{i,0}, \underline{x}_i, \overline{x}_i, b_i, f_i, Q_i, R_i, \omega_i, a_{i,11}, a_{i,12}, \\ &\quad \quad a_{i,22}, \{a_{ij}\}_{j \in \mathcal{N}_i}, b_{i, 1}, b_{i, 2}).
	\end{aligned}
\end{equation}
Note that, to ensure convexity, a second-order parameterization is used with the cost quadratic in $Q_i$ and $R_i$.
The learnable model parameters are initialized as the inaccurate model \eqref{eq:inac_mod}, the cost matrices as $Q_i = \bm{1} \in \mathbb{R}^{n}$ and $R = \frac{1}{2}\bm{1} \in \mathbb{R}^m$, and the local penalty $\omega_i = \omega$.
All other learnable parameters are initialized to zero.

We now compare the distributed second-order approach proposed in this paper (D-SO) against the first-order version (D-FO) \citep{mallick2024multi}, where the first-order update \eqref{eq:fo_local} is used, and the centralized equivalent with second-order updates (C-SO) \citep{gros_data-driven_2020}, where the global MPC problem for the entire network is solved in a centralized manner, and the global update \eqref{eq:global_update} is used. 
For each approach, $\alpha$ is chosen as the highest performing in the set $\{10^{-4}, 10^{-5}, \dots, 10^{-10}\}$, with $\alpha = 10^{-4}$ chosen for both second-order methods, and $\alpha=10^{-8}$ chosen for the first-order method.
For all approaches $T = 15$ and $\mathcal{T}$ is sampled uniformly from a replay buffer of the $100$ most recent transitions.
For regularization, $\sigma_i$ is chosen to render $T\sigma_i - \bm{K}$ non-singular, with the invertibility of $\bm{I}+\bm{G}$ checked before computing \eqref{eq:local_so_update} in the distributed case.
For the distributed approach we use 100 iterations in both ADMM and GAC.
Figure \ref{fig:td_r} shows the distribution of the global TD error and the collective cost across five training instances, to account for randomness.
Figure \ref{fig:states} shows the state and input trajectories across one representative instance.
It can be seen that the behavior and the learning outcome of the distributed second-order approach is similar to that of the centralized second-order approach. 
Agents learn to regulate the state as close to the origin as possible without incurring expensive constraint violations.
Minor differences can be explained by small errors in the primal and dual values from solving the global MPC problem with ADMM that propagate through the parameter updates.
Both second-order approaches significantly outperform the first-order approach, which fails to make significant learning progress.
\begin{figure}
	\centering
	\input{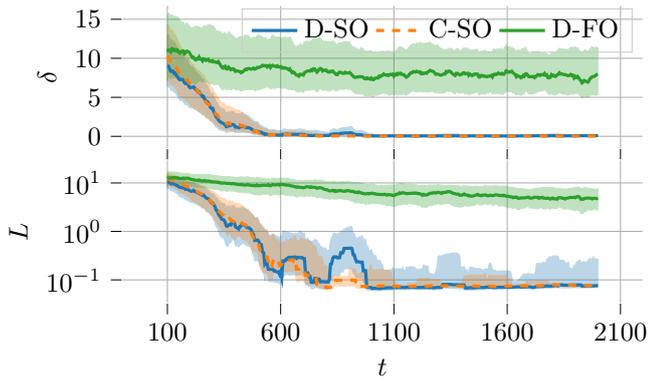}
	\caption{Moving average (100 steps) of TD error (top) and global stage cost (bottom, log scale). Five training instances are shown, with solid lines the median and shaded areas the interval between the 32nd and 68th percentiles.}
	\label{fig:td_r}
\end{figure}
\begin{figure}
	\centering
	\input{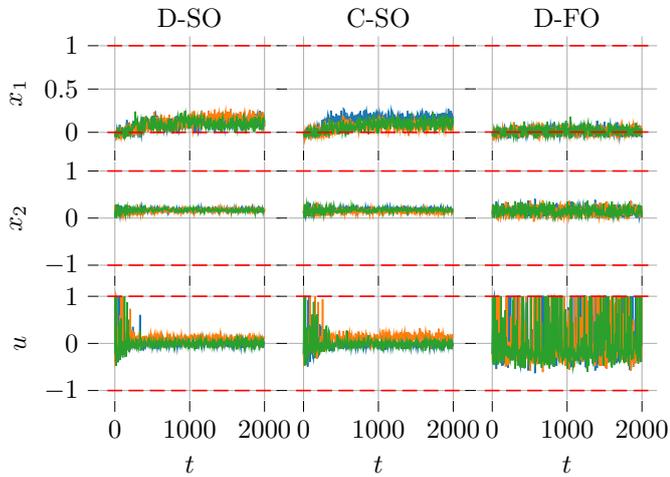}
	\caption{State and action trajectories of agents during a learning instance.}
	\label{fig:states}
\end{figure}

\section{Conclusions}\label{sec:conclusions}
This work has extended the MPC-based distributed Q-learning paradigm to allow for second-order updates, thereby improving learning performance.
In particular, it is shown that leveraging consensus on certain gradient information (the matrix $\bm{C}$) allows agents to reconstruct an update that is equivalent to the corresponding portion of the global learning update.
In simulation, the approach is shown to perform comparatively to centralized second-order MPC-based Q-learning, while being fully distributed, and to outperform distributed first-order MPC-based Q-learning.
 
Future work will look at extending this methodology to the class of policy-based learning algorithms, e.g., policy gradient.

\bibliography{ifacconf}             
                                                  
\end{document}